\documentclass[12pt]{article}
\usepackage{amsthm,amssymb,amsfonts,amsmath,amscd}
\usepackage{graphicx}
\makeatletter \@addtoreset{equation}{section} \makeatother
\newtheorem{proposition}{Proposition}
\newtheorem{theorem}{Theorem}
\newtheorem{lemma}{Lemma}

\def\dfrac{\displaystyle\frac}

\def\intd{\displaystyle\int}

\def\mdet{\mathrm{det}}

\begin{document}
\title{On the correlation function of the characteristic polynomials of the hermitian Wigner ensemble}
\author{ T. Shcherbina\\
 Institute for Low Temperature Physics, Kharkov,
Ukraine. \\E-mail: t\underline{ }shcherbina@rambler.ru
}
\date{}
\maketitle \centerline{{\it 2000 Mathematics Subject Classification.} Primary 15A52; Secondary 15A57}
\begin{abstract}
We consider the asymptotics of the correlation functions of the characteristic polynomials of the
hermitian Wigner matrices $H_n=n^{-1/2}W_n$. We show that for the correlation function of any even order the
asymptotic coincides with this for the GUE up to a
factor, depending only on the forth moment of the common probability law $Q$ of entries
$\Im W_{jk}$, $\Re W_{jk}$, i.e. that the higher moments of $Q$ do not
contribute to the above limit.
\end{abstract}
\section{Introduction}
Characteristic polynomials of random matrices have been actively studied
in the last years. The interest was initially stimulated by the similarity between the
asymptotic behavior of the moments of characteristic polynomials of a random matrix from the Circular Unitary Ensemble
and the moments of the Riemann $\zeta$-function along its critical line (see \cite{K-Sn:00}).
But with the emerging connections to the quantum chaos, integrable systems, combinatorics, representation
theory and others, it has become apparent that the characteristic polynomials of random matrices are also
of independent interest. This motivate the asymptotic study of the moments of characteristic
polynomials for other random matrix ensembles (see e.g. \cite{Me-Nor:01}, \cite{Br-Hi:01}).

In this paper we consider the hermitian Wigner Ensembles with symmetric entries distribution, i.e. hermitian $n\times n$
random matrices
\begin{equation}\label{H}
H_n=n^{-1/2}W_n
\end{equation}
with independent (modulo symmetry) and identically distributed entries $\Re W_{j,k}$
and $\Im W_{j,k}$ such that
\begin{equation}\label{W}
\begin{array}{cc}
\mathbf{E}\{W_{jk}\}=\mathbf{E}\{(W_{jk})^2\}=0,&\quad \mathbf{E}\{|W_{jk}|^2\}=1,\\
\mathbf{E}\{\Re^{2l+1} W_{jk}\}=\mathbf{E}\{\Im^{2l+1} W_{jk}\}=0,&\quad j,k=1,..,n,\quad l\in \mathbb{N}.
\end{array}
\end{equation}
Denote by $\lambda_1^{(n)},\ldots,\lambda_n^{(n)}$ the eigenvalues of
random matrix and define their Normalized Counting Measure
(NCM) as
\begin{equation}  \label{NCM}
N_n(\triangle)=\sharp\{\lambda_j^{(n)}\in
\triangle,j=1,..,n \}/n,\quad N_n(\mathbb{R})=1,
\end{equation}
where $\triangle$ is an arbitrary interval of the real axis. The global regime of the random matrix theory,
centered around the weak convergence of the Normalized Counting Measure of
eigenvalues, is well-studied for many ensembles. It is shown that
$N_n$ converges weakly to a non-random limiting measure $N$ known
as the Integrated Density of States (IDS). The IDS is normalized
to unity and is absolutely continuous in many cases
\begin{equation}  \label{rho}
N(\mathbb{R})=1,\quad N(\triangle)=\displaystyle\int\limits_\triangle
\rho(\lambda)d\,\lambda.
\end{equation}
The non-negative function $\rho$ in (\ref{rho}) is called the limiting
density of states of the ensemble.
In the case of Wigner hermitian ensemble it is well-known (see, e.g.,\cite{Pa:72}) that
\begin{equation}\label{rho_sc}
\rho(\lambda)=\rho_{sc}(\lambda)=\dfrac{1}{2\pi}\sqrt{4-\lambda^2}.
\end{equation}

  The mixed moments or the correlation functions of characteristic polynomials are
\begin{equation}\label{F}
F_{2m}(\Lambda)=\displaystyle\int\limits_{\mathcal{H}_n}\prod\limits_{j=1}^{2m}\mdet(\lambda_j-H)P_n(d\,H_n),
\end{equation}
where $\mathcal{H}_n$ is the space of hermitian $n\times n$ matrices,
\begin{equation}\label{dH}
d\,H_n=\prod\limits_{j=1}^n d\,H_{jj}\prod\limits_{1\le j<k\le n}\Re H_{j,k}\Im H_{j,k}
\end{equation}
is the standard Lebesgues measure on $\mathcal{H}_n$, $P_n(d\,H_n)$ is probability law of the $n\times n$ random matrix $H_n$,
and $\Lambda=\{\lambda_j\}_{j=1}^{2m}$ are real or complex parameters
that may depend on $n$.

We are interested in the asymptotic behavior of (\ref{F}) for matrices (\ref{H}) as $n\to\infty$ for
\begin{equation*}
\lambda_j=\lambda_0+\dfrac{\xi_j}{n\rho_{sc}(\lambda_0)},\quad j=1,..,2m,
\end{equation*}
where $\lambda_0\in (-2,2)$, $\rho_{sc}$ is defined in (\ref{rho_sc}) and
$\widehat{\xi}=\{\xi_j\}_{j=1}^{2m}$ are real number varying in a
compact set $K\subset \mathbb{R}$.

  In the case of hermitian matrix model, i.e. the matrices with
\[
P_n(d\,H_n)=Z_n^{-1}e^{-n\,\hbox{tr}\,V(H_n)}d\,H_n,
\]
where $V$ is a potential function, the asymptotic behavior of (\ref{F}) is known. Using the method of
orthogonal polynomials, it was shown (see \cite{St-Fy:03},\cite{Br-Hi:00}) that
\begin{multline*}
\dfrac{1}{(n\rho(\lambda_0))^{m^2}}F_{2m}\left(\Lambda_0+\widehat{\xi}/(n\rho(\lambda_0))\right)\\
=C_n\dfrac{e^{mnV(\lambda_0)+\alpha_V(\lambda_0)
\sum\limits_{j=1}^{2m}\xi_j}}{\triangle(\xi_1,..,\xi_m)\triangle(\xi_{m+1},..,\xi_{2m})}\mdet
\left\{\dfrac{\sin(\pi(\xi_i-\xi_{m+j}))}{\pi(\xi_i-\xi_{m+j})}
\right\}_{i,j=1}^m(1+o(1)),\,\,n\to\infty,
\end{multline*}
where $\Lambda_0=(\lambda_0,\ldots,\lambda_0)\in \mathbb{R}^{2m}$,
\begin{equation*}
\alpha_V(\lambda)=\dfrac{V^\prime(\lambda)}{2\rho(\lambda)},
\end{equation*}
$\rho$ is a density of (\ref{rho}), $\lambda_0$ is such that $\rho(\lambda_0)>0$ and $\triangle(x_1,\dots,x_m)$ is the
Vandermonde determinants of $x_1,\dots,x_m$.

 Unfortunately, the method of orthogonal polynomials can not be applied to the general case of hermitian Wigner
Ensembles. Thus, to find the asymptotic behavior of (\ref{F}) other methods should be used. In \cite{Got-K:08} Gotze
and Kosters use the exponential generating function to study this behavior for the second moment, i.e.
for the case $m=1$ in (\ref{F}). In this case
it was shown for matrices (\ref{H}) that
\begin{multline*}
\dfrac{1}{n\rho(\lambda_0)}F_2\left(\lambda_0+\xi_1/(n\rho_{sc}(\lambda_0)),
\lambda_0+\xi_2/(n\rho(\lambda_0))\right)\\
=2\pi\exp\{n(\lambda_0^2-2)/2+\alpha(\lambda_0)(\xi_1+\xi_2)+|\kappa_4|\}\dfrac{\sin(\pi(\xi_1-\xi_2))}{\pi(\xi_1-\xi_2)}(1+o(1)),
\end{multline*}
where
\begin{equation}\label{alpha,kap}
\alpha(\lambda)=\dfrac{\lambda}{2\rho_{sc}(\lambda)}, \quad \kappa_4=\mu_4-3/4,
\end{equation}
and $\mu_4$ is the forth moment of the common probability law $Q$ of entries
$\Im W_{jk}$, $\Re W_{jk}$.

In this paper we consider the general case $m\ge 1$ of (\ref{F}) for the random matrices (\ref{H}).
Set
\begin{eqnarray}\label{D_xi}
D^{(n)}(\xi)&=&\frac{1}{n\rho(\lambda_0)}F_{2}\left(\lambda_0+\frac{\xi}{n\rho_{sc}(\lambda_0)},
\lambda_0+\frac{\xi}{n\rho_{sc}(\lambda_0)}\right)\\ \notag
&=&2\pi \exp\left\{\frac{n}{2}(\lambda_0^2-2)+2\alpha(\lambda_0)\xi+\kappa_4\right\}(1+o(1)).
\end{eqnarray}

  The main result of the paper is
\begin{theorem}\label{thm:1}
Let the entries $\Im W_{jk}$, $\Re W_{jk}$ of matrices (\ref{H}) has the
symmetric probability distribution with $4m$ finite moments. Then we have for $m\ge 1$
\begin{multline*}
\lim\limits_{n\to\infty}\dfrac{1}{(n\rho_{sc}(\lambda_0))^{m^2}\prod\limits_{l=1}^{2m}\sqrt{D^{(n)}(\xi_l)}}
F_{2m}\left(\Lambda_0+\widehat{\xi}/(n\rho_{sc}(\lambda_0))\right)\\
=\dfrac{\exp\{m(m-1)\kappa_4(\lambda_0^2-2)^2/2\}}{\pi^{2m(m-1)}\Delta(\xi_1,...,\xi_m)
\Delta(\xi_{m+1},...,\xi_{2m})}\mdet
\left\{\dfrac{\sin(\pi(\xi_i-\xi_{m+j}))}{\pi(\xi_i-\xi_{m+j})}
\right\}_{i,j=1}^m,
\end{multline*}
where $F_{2m}$ and $\rho_{sc}(\lambda)$ are defined in (\ref{F}) and (\ref{rho_sc}),
$\Lambda_0=(\lambda_0,\ldots,\lambda_0)\in \mathbb{R}^{2m}$,
$\lambda_0\in(-2,2)$, $\widehat{\xi}=\{\xi_j\}_{j=1}^{2m}$, and $\alpha(\lambda)$ and
$\kappa_4$ are defined in (\ref{alpha,kap}).
\end{theorem}
The theorem shows that the above limit for the mixed moments of characteristic polynomials
for random matrices (\ref{H}) coincide with those for the GUE up to a
factor, depending only on the forth moment of the common probability law $Q$ of entries
$\Im W_{jk}$, $\Re W_{jk}$, i.e. that the higher moments of $Q$ do not
contribute to the above limit. This is a manifestation of universality of the limit,
that can be composed with universality of the local bulk regime for Wigner matrices (see \cite{ErTao:09}).

The paper is organized as follows. In Section $2$ we obtain a convenient integral representation for $F_{2m}$
in the case of symmetric probability distribution of entries with $4m$ finite moments by using the integration
over the Grassmann variables and Harish Chandra/Itzykson-Zuber formula for integrals
over the unitary group. In Section $3$ we prove Theorem \ref{thm:1} by applying the steepest descent method
to the integral representation.

We denote by $C, C_1$, etc. and $c, c_1$, etc. various $n$-independent constants below, which
can be different in different formulas. Integrals without limits denote the integrals over whole real
axis.

\section{The integral representation.}
In this section we obtain the integral representation for the correlation functions $F_{2m}$ of (\ref{F})
of characteristic polynomials. To this end we use the integration
over the Grassmann variables. The integration was introduced by Berezin
and widely used in the physics literature (see \cite{Ber} and \cite{Ef}).
For the reader convenience we give an outline of this technique here.

Let us consider the two sets of formal variables
$\{\psi_j\}_{j=1}^n,\{\overline{\psi}_j\}_{j=1}^n$, which satisfy the following anticommutation
conditions
\[
\psi_j\psi_k+\psi_k\psi_j=\overline{\psi}_j\psi_k+\psi_k\overline{\psi}_j=\overline{\psi}_j\overline{\psi}_k+
\overline{\psi}_k\overline{\psi}_j=0,\quad j,k=1,..,n.
\]
 In particular, for $k=j$ we obtain
\[
\psi_j^2=\overline{\psi}_j^2=0.
\]
These two sets of variables $\{\psi_j\}_{j=1}^n$ and
$\{\overline{\psi}_j\}_{j=1}^n$ generate the Grassmann algebra
$\Lambda$. Taking into account that $\psi_j^2=0$, we have that all elements
of $\Lambda$ are polynomials of $\{\psi_j\}_{j=1}^n$ and
$\{\overline{\psi}_j\}_{j=1}^n$. We can also define functions of Grassmann variables. Let $\chi$ be
an element of $\Lambda$. For any analytical function $f$ by
$f(\chi)$ we mean the element of $\Lambda$ obtained by
substituting $\chi$ in the Taylor series of $f$ near zero. Since
$\chi$ is a polynomial of $\{\psi_j\}_{j=1}^n$, $\{\overline{\psi}_j\}_{j=1}^n$,
there exists such $l$ that $\chi^l=0$, and hence the series
terminates after a finite number of terms and so $f(\chi)\in
\Lambda$.

Following Berezin \cite{Ber}, we define the operation of
integration with respect to the anticommuting variables in a formally
way:
\[
\intd d\,\psi_j=\intd d\,\overline{\psi}_j=0,\quad \intd
\psi_jd\,\psi_j=\intd \overline{\psi}_jd\,\overline{\psi}_j=1.
\]
This definition can be extend on the general element of $\Lambda$ by
the linearity. A multiple integral is defined to be repeated
integral. The "differentials" $d\,\psi_j$ and
$d\,\overline{\psi}_k$ anticommute with each other and with the
variables $\psi_j$ and $\overline{\psi}_k$.

Therefore, if
$$f(\chi_1,\ldots,\chi_m)=a_0+\sum\limits_{j_1=1}^m
a_{j_1}\chi_{j_1}+\sum\limits_{j_1<j_2}a_{j_1j_2}\chi_{j_1}\chi_{j_2}+
\ldots+a_{1,2,\ldots,m}\chi_1\ldots\chi_m,
$$
then
\[
\intd f(\chi_1,\ldots,\chi_m)d\,\chi_m\ldots d\,\chi_1=a_{1,2,\ldots,m}.
\]

   Let $A$ be an ordinary hermitian matrix. The following Gaussian integral is well-known
\begin{equation}\label{G_C}
\intd \exp\left\{-\sum\limits_{j,k=1}^nA_{j,k}z_j\overline{z}_k\right\}
\prod\limits_{j=1}^n\dfrac{d\,\Re z_jd\,\Im z_j}{\pi}=\dfrac{1}{\mdet A}.
\end{equation}
One of the important formulas of the Grassmann variables theory is
an analog of formula (\ref{G_C}) for Grassmann algebra (see \cite{Ber}):
\begin{equation}\label{G_Gr}
\int \exp\left\{\sum\limits_{j,k=1}^nA_{j,k}\overline{\psi}_j\psi_k\right\}
\prod\limits_{j=1}^nd\,\overline{\psi}_jd\,\psi_j=\mdet A,
\end{equation}
where $\{\psi_j\}_{j=1}^n$ and $\{\overline{\psi}_j\}_{j=1}^n$ are the Grassmann variables.
Besides, we have
\begin{equation}\label{G_Gr_1}
\int \prod\limits_{p=1}^q\overline{\psi}_{l_p}\psi_{s_p}\exp\left\{\sum\limits_{j,k=1}^nA_{j,k}\overline{\psi}_j
\psi_k\right\}\prod\limits_{j=1}^nd\,\overline{\psi}_jd\,\psi_j=\mdet A_{l_1,..,l_q;s_1,..,s_q},
\end{equation}
where $A_{l_1,..,l_q;s_1,..,s_q}$ is a $(n-q)\times (n-q)$ minor of the matrix $A$ without rows $l_1,..,l_q$ and
columns $s_1,..,s_q$.

\subsection{Asymptotic integral representation for $F_2$.}
In this subsection we obtain the asymptotic integral representation of (\ref{F}) for $m=1$.
The corresponding asymptotic formula was obtained in \cite{Got-K:08} by using the
exponential generating function. We give here a detailed proof based on the Grassmann integration to
show the basic ingredients of our technique. The technique will be elaborated in the next subsection
to obtain the asymptotic integral representation of (\ref{F}) for $m>1$.

Set
\begin{equation}\label{D_2}
D_2=\prod\limits_{l=1}^{2}\sqrt{D^{(n)}(\xi_l)},
\end{equation}
where $D^{(n)}(\xi)$ is defined in (\ref{D_xi}).
Note also that
\begin{equation}\label{D_int}
\sqrt{D^{(n)}(\xi)}=e^{\alpha(\lambda_0)\xi+\kappa_4/2}
\sqrt{\dfrac{n(4-\lambda_0^2)}{2}}\displaystyle\int\left|(t-i\lambda_0/2)^ne^{-\frac{n}{2}
(t+i\lambda_0/2)^2}\right|d\,t(1+o(1))
\end{equation}
as $n\to\infty$.

Using (\ref{G_Gr}), we obtain from (\ref{F})
\begin{equation}\label{ne_usr_2}
\begin{array}{c}
D_2^{-1}F_2(\Lambda)=D_2^{-1}{\bf
E}\left\{\displaystyle\int e^{\sum\limits_{l=1}^2\sum\limits_{j,k=1}^n(\lambda_l-H)_{j,k}
\overline{\psi}_{jl}\psi_{kl}}\prod\limits_{r=1}^2\prod\limits_{q=1}^nd\,\overline{\psi}_{qr}d\,\psi_{qr}\right\}\\
=D_2^{-1}{\bf E}\left\{\displaystyle\int e^{\sum\limits_{s=1}^2\lambda_s\sum\limits_{p=1}^n
\overline{\psi}_{ps}\psi_{ps}}
\exp\left\{-\sum\limits_{j<k}\dfrac{\Re w_{j,k}}{\sqrt{n}}\sum\limits_{l=1}^2(\overline{\psi}_{jl}\psi_{kl}
+\overline{\psi}_{kl}\psi_{jl})\right.\right.\\
\left.\left.-\sum\limits_{j<k}\dfrac{i\Im w_{j,k}}{\sqrt{n}}\sum\limits_{l=1}^2(\overline{\psi}_{jl}\psi_{kl}-\overline{\psi}_{kl}
\psi_{jl})-\sum\limits_{j=1}^n\frac{w_{jj}}{\sqrt{n}}\sum\limits_{l=1}^2\overline{\psi}_{jl}\psi_{jl}\right\}
\prod\limits_{r=1}^2\prod\limits_{q=1}^nd\,\overline{\psi}_{qr}d\,\psi_{qr}\right\},
\end{array}
\end{equation}
where $\{\psi_{jl}\}_{j,l=1}^{n \, 2}$ are the
Grassmann variables ($n$ variables for each determinant in (\ref{F})).
Denote
\begin{eqnarray}\label{hi_2}
\chi^+_{j,k}&=&\sum\limits_{l=1}^2(\overline{\psi}_{jl}\psi_{kl}
+\overline{\psi}_{kl}\psi_{jl}),\quad \chi^-_{j,k}=\sum\limits_{l=1}^2(\overline{\psi}_{jl}\psi_{kl}
-\overline{\psi}_{kl}\psi_{jl}),\,\,j\ne k,\\ \notag
\chi^+_{j,j}&=&\sum\limits_{l=1}^2\overline{\psi}_{jl}\psi_{jl},\quad j,k=1,..,n.
\end{eqnarray}
Using that $(\chi^{\pm}_{j,k})^s=0$ for $s>4$, $j,k=1,..,n$ (since $\overline{\psi}_{js}^2=\psi_{js}^2=0$ for
any $j=1,..,n$, $s=1,2$), we expand the second exponent under the integral in (\ref{ne_usr_2})
into the series and integrate with respect
to the measure (\ref{W}). We get then
\begin{equation}\label{usr_2}
\begin{array}{c}
D_2^{-1}F_2(\Lambda)
=D_2^{-1}\displaystyle
\int e^{\sum\limits_{s=1}^2\lambda_s\sum\limits_{p=1}^n\overline{\psi}_{ps}\psi_{ps}}
\displaystyle\prod\limits_{j<k}\left(1+\dfrac{(\chi^+_{j,k})^2}{4n}
+\dfrac{\mu_4}{4!n^2}(\chi^+_{j,k})^4\right)\\
\displaystyle\prod\limits_{j<k}\left(1-\dfrac{(\chi^-_{j,k})^2}{4n}
+\dfrac{\mu_4}{4!n^2}(\chi^-_{j,k})^4\right)
\displaystyle\prod\limits_{j=1}^n\left(1+\dfrac{1}{2n}(\chi^+_{j,j})^2\right)\prod\limits_{r=1}^2
\prod\limits_{q=1}^nd\,\overline{\psi}_{qr}
d\,\psi_{qr},
\end{array}
\end{equation}
where $\mu_4$ is $4$-th moment of the common probability law $Q$ of the entries
$\Im W_{jk}$, $\Re W_{jk}$ of (\ref{W}). Note that
\begin{eqnarray*}
1\pm\dfrac{1}{4n}(\chi^\pm_{j,k})^2
+\dfrac{\mu_4}{4!n^2}(\chi^\pm_{j,k})^4&=&\exp\left\{\pm\dfrac{1}{4n}(\chi^\pm_{j,k})^2
+\dfrac{\kappa_4}{4!n^2}(\chi^\pm_{j,k})^4\right\},\,\,\,j\ne k,\\
1+\dfrac{1}{2n}(\chi^+_{j,j})^2
&=&\exp\{\dfrac{1}{2n}(\chi^+_{j,j})^2\},\quad j,k=1,..,n,
\end{eqnarray*}
where $\kappa_4$ is defined in (\ref{alpha,kap}).
Thus, (\ref{usr_2}) yields
\begin{equation}\label{ne_fact_2}
D_2^{-1}F_2(\Lambda)=D_2^{-1}
\displaystyle\int e^{\sum\limits_{s=1}^2
\lambda_s\sum\limits_{p=1}^n\overline{\psi}_{ps}\psi_{ps}-\frac{1}{2n}\sigma_1
+\frac{\kappa_4}{n^2}\sigma_2}\prod\limits_{r=1}^2\prod\limits_{q=1}^nd\,\overline{\psi}_{qr}d\,\psi_{qr},
\end{equation}
where
\begin{eqnarray}\label{sigma_2}
\sigma_1&=&-\dfrac{1}{2}\sum\limits_{j<k}\left((\chi^+_{j,k})^2-(\chi^-_{j,k})^2\right)-\sum\limits_{j=1}^n(\chi^+_{j,j})^2\\
\notag
&=&\sum\limits_{l=1}^2\left(\sum\limits_{j=1}^n\overline{\psi}_{jl}\psi_{jl}\right)^2+
2\sum\limits_{1\le l<s\le 2}\left(\sum\limits_{j=1}^n\overline{\psi}_{jl}\psi_{js}\cdot
\sum\limits_{k=1}^n\overline{\psi}_{ks}\psi_{kl}\right),\\
\notag
\sigma_2&=&\dfrac{1}{4!}\sum\limits_{j<k}\left((\chi^+_{j,k})^4+(\chi^-_{j,k})^4\right)=\left(\sum\limits_{j=1}^n\overline{\psi}_{j1}\overline{\psi}_{j2}\psi_{j1}\psi_{j2}\right)^2.
\end{eqnarray}
Now we use the formulas
\begin{eqnarray}\label{gauss}
\sqrt{\dfrac{\pi}{a}}\exp\{ab^2\}&=&\displaystyle\int\exp\{-ax^2-2abx\}d\,x,\\ \notag
\dfrac{\pi}{a}\exp\{abc\}&=&\displaystyle\int\exp\{-a\overline{u}u-ab\overline{u}-acu\}d\,\Re u d\,\Im u,
\end{eqnarray}
where $b,c$ are complex numbers or even Grassmann variables (i.e. sums of the products of even number
of Grassmann variables), and $a$ is a positive number. For the case of even Grassmann variables this
formulas can be obtained by expanding the exponent into the series and integrating of each term.
Therefore, (\ref{sigma_2}) -- (\ref{gauss}) imply
\begin{equation}\label{s1_2}
\begin{array}{c}
\exp\left\{-\dfrac{1}{2n}\sigma_1\right\}=\dfrac{n^2}{2\pi^2}\int\limits_{\mathcal{H}_2}\displaystyle
\exp\left\{-\dfrac{n}{2}\left(\sum\limits_{q=1}^2\tau_q^2
+2\sum\limits_{1\le a<b\le 2}\overline{u}_{ab}u_{ab}\right)\right\}\\
\prod\limits_{j=1}^n
\exp\left\{\sum\limits_{p=1}^2i\tau_p\overline{\psi}_{jp}\psi_{jp}
+\sum\limits_{1\le c<d\le 2}\left(iu_{cd}\overline{\psi}_{jc}\psi_{jd}+
i\overline{u}_{cd}\overline{\psi}_{jd}\psi_{jc}\right)\right\}d\,Q,
\end{array}
\end{equation}
where
\begin{equation}\label{Q_2}
Q=
\left(
\begin{array}{cc}
\tau_1&u_{12}\\
\overline{u}_{12}&\tau_2\\
\end{array}\right),
\end{equation}
$\mathcal{H}_2$ is the space of $2\times 2$ hermitian matrices and $d\,Q$ is given in (\ref{dH}) for $n=2$.
Write the formula
\begin{equation}\label{s2_2}
\exp\left\{\dfrac{\kappa_4}{n^2}\sigma_2\right\}
=\sqrt{\dfrac{|\kappa_4|}{\pi}}\displaystyle\int \exp\{-|\kappa_4|p^2\}\\
\prod\limits_{j=1}^n
\exp\left\{\dfrac{2p\,\varepsilon(\kappa_4)}{n}\overline{\psi}_{j1}\overline{\psi}_{j2}\psi_{j1}\psi_{j2}
\right\}d\,p
\end{equation}
with
\begin{equation}\label{eps}
\varepsilon(x)=\left\{
\begin{array}{ll}
x,& x>0,\\
-ix,& x<0.
\end{array}
\right.
\end{equation}

Substituting (\ref{s1_2}) -- (\ref{s2_2}) in (\ref{ne_fact_2}) and using (\ref{G_Gr}) -- (\ref{G_Gr_1}) we can
integrate in (\ref{ne_fact_2}) over the Grassmann variables. We obtain
\begin{equation}\label{F_int11_2}
D_2^{-1}F_2(\Lambda)=Z_2
\displaystyle\int\,d\,p \displaystyle\int\limits_{\mathcal{H}_2}\,d\,Q e^{-\frac{n}{2}\hbox{tr}\,Q^2-|\kappa_4|p^2}
\left(\mdet(Q-i\Lambda)+\dfrac{2p\,\varepsilon(\kappa_4)}{n}\right)^n,
\end{equation}
where $Q$ is defined in (\ref{Q_2}) and
\begin{equation}\label{lambda_2}
\Lambda=
\left(
\begin{array}{cc}
\lambda_1&0\\
0&\lambda_2\\
\end{array}
\right),\quad
Z_2=\dfrac{(-1)^n n^2D_2^{-1}}{2\pi^2\sqrt{\pi|\kappa_4|^{-1}}}.
\end{equation}
Recall that we are interested in $\Lambda=\Lambda_0+\widehat{\xi}/n\rho_{sc}(\lambda_0)$, where
$\Lambda_0=\hbox{diag}\{\lambda_0,\lambda_0\}$ and
$\widehat{\xi}=\hbox{diag}\{\xi_1,\xi_{2}\}$.
Change variables to $\tau_j-i\lambda_0/2-i\xi_j/n\rho_{sc}(\lambda_0)\to \tau_j$, $j=1,2$ and
note that we can move the integration with respect to $\tau_j$ from line $\Im z=\lambda_0/2+\xi_j/n\rho_{sc}(\lambda_0)$
back to the real axis. Indeed, consider the contour $C_{jR}$, which is the rectangle with vertices at $(-R,0)$,
$(-R,\lambda_0/2+\xi_j/n)$, $(R,\lambda_0/2+\xi_j/n)$ and $(R,0)$.
Since the integrand in (\ref{F_int11_2}) is analytic in $\{\tau_j\}_{j=1}^2$, the integral
with respect to $\tau_j$ of this function over $C_{jR}$ is equal to $0$. Besides, the integral over the segments
of lines $\Im z=\pm R$ tends to $0$ as $R\to\infty$, since the integrand in (\ref{F_int11_2})
is a polynomial of $\tau_j$ multiplied by $\exp\{-n\tau_j^2/2\}$. Thus, setting $R\to \infty$, we obtain
that the integral with respect to $\tau_j$ over the line $\Im z=\lambda_0/2+\xi_j/n$ is equal to
the integral over the real axis.
Hence, we obtain in new variables
\begin{equation}\label{F_int1_2}
\begin{array}{c}
D_2^{-1}F_2(\Lambda)=Z_2
\displaystyle\int d\,p\displaystyle\int\limits_{\mathcal{H}_2}d\,Q  e^{-\frac{n}{2}
\hbox{tr}\,(Q+\frac{i\Lambda_0}{2})^2-\frac{i}{\rho_{sc}(\lambda_0)}\hbox{tr}(Q+\frac{i\Lambda_0}{2})\widehat{\xi}
-|\kappa_4|p^2-\frac{1}{2n}\hbox{tr}\,\frac{\widehat{\xi}^2}{\rho_{sc}(\lambda_0)^2}}\\
\times\left(\mdet(Q-\frac{i\Lambda_0}{2})+\dfrac{2p\,\varepsilon(\kappa_4)}{n}\right)^n
=Z_2\displaystyle\int d\,p\displaystyle\int\limits_{\mathcal{H}_2} d\,Q \exp\left\{-\frac{i}{\rho_{sc}(\lambda_0)}
\hbox{tr}(Q+\frac{i\Lambda_0}{2})\widehat{\xi}\right.\\
\left.-|\kappa_4|p^2-\frac{1}{2n}\hbox{tr}\,\frac{\widehat{\xi}^2}{\rho_{sc}(\lambda_0)^2}\right\}
 \mu_n(Q)\left(1+\dfrac{2p\,\varepsilon(\kappa_4)}{n\mdet(Q-i\Lambda_0/2)}\right)^n,
\end{array}
\end{equation}
where $Q$ is again the hermitian (see (\ref{Q_2})) and
\begin{equation}\label{mu}
\mu_n(Q)=\mdet^n(Q-i\Lambda_0/2)e^{-\frac{n}{2}
\hbox{tr}\,(Q+i\Lambda_0/2)^2}.
\end{equation}
Let $q_1,q_2$ be the eigenvalues of $Q$. Set
\begin{equation}\label{Om_n}
\begin{array}{c}
\widetilde{\Omega}_{n}=\{(Q,p):a\le |q_l-i\lambda_0/2|\le A, l=1,2,\,|p|\le \log n\}, \\
\widetilde{\Omega}_n^Q=\{Q\in \mathcal{H}_2: a\le |q_l-i\lambda_0/2|\le A\}.
\end{array}
\end{equation}
for sufficiently small $a$ and sufficiently big $A$ (note that if $|\lambda_0|\ge \delta$, then
$|q_l-i\lambda_0/2|\ge \delta^2/4$ and
we can omit the first inequality in (\ref{Om_n})).
Note that the integral in (\ref{F_int1_2}) over the domain $\max\limits_{l=1,2}|q_l|\ge A$
 is $O(e^{-nA^2/4})$, $A\to\infty$ and the integral over the domain $\min\limits_{l=1,2} |q_l|\le a$ is
$O(e^{-n\log a^{-1}})$, $a\to 0$. If $a \le |q_l-i\lambda_0/2|\le A$ and $|p|\ge \log n$,
then according to (\ref{D_2}), (\ref{D_int})  and (\ref{Q_2}), the corresponding integral is bounded by
\begin{equation}\label{bound_vne}
Z_2\displaystyle\int\limits_{\widetilde{\Omega}_n^Q} |\mu_n(Q)| d\,Q\displaystyle\int\limits_{|p|\ge \log n}(1+Cp/n)^n
e^{-|\kappa_4|p^2}d\,p=O(e^{-C\log^2n }),
\end{equation}
and we can write
\begin{equation}\label{F_int2_2}
\begin{array}{c}
D_2^{-1}F_{2}(\Lambda)=Z_2\displaystyle\int\limits_{\widetilde{\Omega}_{n}}
e^{-i\hbox{tr}(Q+\frac{i\Lambda_0}{2})\frac{\widehat{\xi}}{\rho_{sc}(\lambda_0)}-|\kappa_4|p^2
-\frac{1}{2n}\hbox{tr}\,\frac{\widehat{\xi}^2}{\rho_{sc}(\lambda_0)^2}
+2p\,\varepsilon(\kappa_4)\,\mdet^{-1}(Q-\frac{i\Lambda_0}{2})}\\
\times\mu_n(Q)\left(1+
f_n(\mdet(Q-i\Lambda_0/2),p)\right)d\,p\,d\,Q+O(e^{-c\log^2 n}),
\end{array}
\end{equation}
where
\begin{equation}\label{f_n_2}
f_n(\mdet(Q-i\Lambda_0/2),p)=e^{-2p\,\varepsilon(\kappa_4)\,\mdet^{-1}(Q-\frac{i\Lambda_0}{2})}
\left(1+\dfrac{2p\,\varepsilon(\kappa_4)}{n\,\mdet (Q-\frac{i\Lambda_0}{2})}\right)^n-1.
\end{equation}
Note that $f_n$ is an analytic function of $p$ and entries of $Q$, and we have on $\widetilde{\Omega}_n$
\begin{equation}\label{bound_f_2}
|f_n(\mdet(Q-i\Lambda_0/2),p)|\le \dfrac{\log^k n}{n},
\end{equation}
where $k$ is independent of $n$. It is easy to check that
\begin{equation*}
I:=\displaystyle\int\limits_{|p|\le \log n}e^{-|\kappa_4|p^2+2p\,\varepsilon(\kappa_4)\,
\mdet^{-1}(Q-\frac{i\Lambda_0}{2})}d\,p=
\sqrt{\dfrac{\pi}{|\kappa_4|}}e^{\kappa_4\,
\mdet^{-2}(Q-i\Lambda_0/2)}+O(e^{-c\log^2 n}),
\end{equation*}
and we obtain that $|I|>C_3>0$ on $\widetilde{\Omega}_n$ (see (\ref{Om_n})).
Thus, (\ref{F_int2_2}) yields
\begin{equation}\label{F_int2}
\begin{array}{c}
D_2^{-1}F_{2}(\Lambda)=\dfrac{n^2D_2^{-1}}
{(-1)^n2\pi^{2}}\int\limits_{\widetilde{\Omega}_n^Q} \mu_n(Q)\exp\left\{-
i\hbox{tr}(Q+i\Lambda_0/2)\widehat{\xi}/\rho_{sc}(\lambda_0)\right.\\ \notag
\left.+\kappa_4\,\mdet^{-2}(Q-i\Lambda_0/2)\right\}
\left(1+f^{(1)}_n(\mdet(Q-i\Lambda_0/2))\right)d\,Q+O(e^{-c\log^2n}),
\end{array}
\end{equation}
where
\begin{multline}\label{f-1_2}
f^{(1)}_n\left(\mdet\left(Q-i\Lambda_0/2\right)\right)=e^{-\frac{1}{2n}\hbox{tr}\,\frac{\widehat{\xi}^2}{\rho_{sc}(\lambda_0)^2}}-1\\
+I^{-1}e^{-\frac{1}{2n}\hbox{tr}\,\frac{\widehat{\xi}^2}{\rho_{sc}(\lambda_0)^2}}\displaystyle\int\limits_{|p|\le\log n} e^{-|\kappa_4|p^2+2p\,\varepsilon(\kappa_4)\,
\mdet^{-1}(Q-i\Lambda_0/2)}
f_n(\mdet(Q-i\Lambda_0/2),p)d\,p.
\end{multline}
According to (\ref{bound_f_2}), we get that $f^{(1)}_n(\mdet(Q-i\Lambda_0/2))$ is
analytic in elements of $Q$ on $\widetilde{\Omega}_n^Q$ and
\begin{equation}\label{bound_f_1}
 |f^{(1)}_n(\mdet(Q-i\Lambda_0/2))|\le \log^k n/n,
\end{equation}
where $k$ is independent of $n$.

Let us change variables to $Q=U^{*}T U$, where $U$ is a unitary matrix and
$T=\hbox{diag}\{t_1,t_2\}$. Then $d\,Q$ of (\ref{dH}) for $n=2$ transforms
to $(t_1-t_2)^2d\,t_1\,d\,t_2 d\,\mu(U)$, where
$\mu(U)$ is the normalized to unity Haar measure on the unitary group
$U(2)$ (see e.g. \cite{Me:91}, Section 3.3). Hence, since functions $\mdet (Q-i\Lambda_0/2)$
and $\hbox{tr}\,(Q+i\Lambda_0/2)^2$
are unitary invariant, (\ref{F_int2}) implies
\begin{eqnarray}\label{F_int4_2}
D_2^{-1}F_{2}(\Lambda)&=&\dfrac{n^2(-1)^n}
{2\pi^{2}D_2}\int\limits_{U(2)}d\mu(U)\int\limits_{L_a^A\times L_a^A}d\,t_1d\,t_2
\prod\limits_{l=1}^{2}\left(t_l-i\lambda_0/2\right)^n\\ \notag
&\times&e^{-\frac{n}{2}\sum\limits_{s=1}^{2}(t_s+\frac{i\lambda_0}{2})^2-
\hbox{tr}\,U^*(T+\frac{i\Lambda_0}{2})U\frac{i\widehat{\xi}}{\rho_{sc}(\lambda_0)}+
\kappa_4\prod\limits_{r=1}^{2}(t_r-\frac{i\lambda_0}{2})^{-2}}\\ \notag
&\times&\left(1+f^{(1)}_n(\mdet(T-\frac{i\Lambda_0}{2}))\right)
+O(e^{-c\log^2n}),
\end{eqnarray}
where
\begin{equation}\label{L}
L_a^A=\{t\in \mathbb{R}:a\le |t-i\lambda_0/2|\le A\}.
\end{equation}
The integral over the unitary group $U(2)$ can be computed using the
well-known Harish Chandra/Itsykson-Zuber formula (see e.g. \cite{Me:91}, Appendix 5)
\begin{proposition}\label{p:Its-Z}
Let $A$ be the normal $n\times n$ matrix with distinct eigenvalues $\{a_i\}_{i=1}^n$ and
$B=\hbox{diag}\{b_1,\ldots,b_n\}$. Then we have
\begin{multline}\label{Its-Zub}
\int\limits_{U(n)}\int \exp\{-\dfrac{1}{2}\hbox{tr} (A-U^*BU)^2\} \triangle^2(B)f(B)d\,Ud\,B\\
=\pi^{n/2}
\int \exp\{-\dfrac{1}{2}\hbox{tr} (a_j-b_j)^2\}\dfrac{\triangle(B)}{\triangle(A)} f(b_1,\ldots,b_n)
d\,B,
\end{multline}
where $f(B)$ is any symmetric function of $\{b_j\}_{j=1}^n$, $d\,B=\prod\limits_{j=1}^nd\,b_j$ and $\triangle(A)$, $\triangle(B)$ are Vandermonde determinants
for the eigenvalues
$\{a_i\}_{i=1}^n$, $\{b_i\}_{i=1}^n$ of $A$ and $B$.
\end{proposition}
Hence, we obtain finally from (\ref{F_int4_2})
\begin{equation}\label{F_int_2}
\begin{array}{c}
D_2^{-1}F_{2}(\Lambda)=\dfrac{i\rho_{sc}(\lambda_0)n^2}
{2\pi(-1)^nD_2}\int\limits_{L_a^A\times L_a^A} \prod\limits_{l=1}^2(t_l-i\lambda_0/2)^n
e^{-\frac{n}{2}\sum\limits_{l=1}^{2}(t_l+\frac{i\lambda_0}{2})^2-\sum\limits_{l=1}^{2}\frac{i\xi_l}{\rho_{sc}(\lambda_0)}
(t_l+\frac{i\lambda_0}{2})}\\
\dfrac{t_1-t_2}{\xi_1-\xi_2}e^{\kappa_4(t_1-i\lambda_0/2)^{-2}(t_2-i\lambda_0/2)^{-2}}
\left(1+f^{(2)}_n(T)\right)
d\,t_1\,d\,t_2+O(e^{-c\log^2n}),
\end{array}
\end{equation}
where $L_a^A$ is defined in (\ref{L}) and $f^{(2)}_n(T)=f^{(1)}_n(\mdet(T-i\Lambda_0/2))$ is an
analytic function bounded by $\log^kn/n$ if $t_l\in L_a^A$, $l=1,2$.

  This asymptotic integral representation is used in the section 3 to prove the theorem for $m=1$.

\subsection{Asymptotic integral representation for $F_{2m}$.}
Set
\begin{equation}\label{D_2m}
D_{2m}=\prod\limits_{l=1}^{2m}\sqrt{D^{(n)}(\xi_l)},
\end{equation}
where $D^{(n)}(\xi)$ is defined in (\ref{D_xi}).
Using (\ref{G_Gr}), we obtain from (\ref{F}) (cf. (\ref{ne_usr_2}))
\begin{equation}\label{ne_usr_m}
\begin{array}{c}
D^{-1}_{2m}F_{2m}(\Lambda)=D^{-1}_{2m}{\bf
E}\left\{\displaystyle\int e^{\sum\limits_{l=1}^{2m}\sum\limits_{j,k=1}^n(\lambda_l-H)_{j,k}
\overline{\psi}_{jl}\psi_{kl}}\prod\limits_{r=1}^{2m}\prod\limits_{q=1}^nd\,\overline{\psi}_{qr}d\,\psi_{qr}\right\}\\
=D^{-1}_{2m}
{\bf E}\left\{\displaystyle\int e^{\sum\limits_{s=1}^{2m}\lambda_s\sum\limits_{p=1}^n\overline{\psi}_{ps}\psi_{ps}
-
\sum\limits_{j<k}\frac{\Re w_{j,k}}{\sqrt{n}}\sum\limits_{l=1}^{2m}(\overline{\psi}_{jl}
\psi_{kl}+\overline{\psi}_{kl}\psi_{jl})}\right.\\
\left.e^{-\sum\limits_{j<k}\frac{i\Im w_{j,k}}{\sqrt{n}}\sum\limits_{l=1}^{2m}
(\overline{\psi}_{jl}\psi_{kl}-\overline{\psi}_{kl}
\psi_{jl})-\sum\limits_{j=1}^n\frac{w_{jj}}{\sqrt{n}}\sum\limits_{l=1}^{2m}\overline{\psi}_{jl}\psi_{jl}}
\prod\limits_{r=1}^{2m}\prod\limits_{q=1}^nd\,\overline{\psi}_{qr}d\,\psi_{qr}\right\},
\end{array}
\end{equation}
where $\{\psi_{jl}\}_{j,l=1}^{n, \, 2m}$ are the
Grassmann variables ($n$ variables for each determinant).
As in (\ref{hi_2}) we denote
\begin{eqnarray}\label{hi_m}
\chi^+_{j,k}&=&\sum\limits_{l=1}^{2m}(\overline{\psi}_{jl}\psi_{kl}
+\overline{\psi}_{kl}\psi_{jl}),\quad \chi^-_{j,k}=\sum\limits_{l=1}^{2m}(\overline{\psi}_{jl}\psi_{kl}
-\overline{\psi}_{kl}\psi_{jl}),\quad j\ne k,\\ \notag
\chi^+_{j,j}&=&\sum\limits_{l=1}^{2m}\overline{\psi}_{jl}\psi_{jl},\quad j,k=1,\ldots,n.
\end{eqnarray}
Using that $(\chi^{\pm}_{j,k})^s=0$ for $s>4m$, $j,k=1,..,2m$ (since $\overline{\psi}_{jl}^2=\psi_{jl}^2=0$ for any $j=1,..,n$,
$l=1,..,2m$), we expand the exponent under the integral in (\ref{ne_usr_m})
into the series and integrate with respect
to the measure (\ref{W}).  We get then similarly to (\ref{ne_fact_2})
\begin{equation}\label{ne_fact_m}
D^{-1}_{2m}F_{2m}(\Lambda)=D^{-1}_{2m}
\displaystyle\int e^{\sum\limits_{s=1}^{2m}
\lambda_s\sum\limits_{k=1}^n\overline{\psi}_{ks}\psi_{ks}-\frac{1}{2n}\sigma_1
+\sum\limits_{p=2}^{2m}\frac{\kappa_{2p}}{n^p}\sigma_p}\prod\limits_{r=1}^{2m}\prod\limits_{q=1}^n
d\,\overline{\psi}_{qr}d\,\psi_{qr},
\end{equation}
where $\kappa_{2p}$ is cumulants of the probability distribution of entries $\Re w_{jk}$, $\Im w_{jk}$
of (\ref{W}), i.e. the coefficients in the expansion
\[
l(t):=\log \mathbf{E}\{e^{it\Re w_{jk} }\}=\sum_{q=0}^{s}
\frac{\kappa_{q}}{q!}(it)^{q}+o(t^{s}),\quad t\rightarrow 0.
\]
The function $\sigma_1$ in (\ref{ne_fact_m}) is the same as in (\ref{sigma_2}) (but with $\chi^\pm_{j,k}$
of (\ref{hi_m}) and the sums from 1 to $2m$ instead of from 1 to 2),
\begin{eqnarray}\label{sigma_m}
\sigma_2&=&\dfrac{1}{4!}\sum\limits_{j<k}((\chi^+_{j,k})^4+(\chi^-_{j,k})^4)+
\dfrac{2}{4!}\sum\limits_{j=1}^n(\chi^+_{j,j})^4\\ \notag
&=&2\sum\limits_{l_1<l_2<s_1<s_2}\sum\limits_{j=1}^n\overline{\psi}_{jl_1}\overline{\psi}_{jl_2}\overline{\psi}_{js_1}
\overline{\psi}_{js_2}\cdot
\sum\limits_{k=1}^n\psi_{kl_1}\psi_{kl_2}\psi_{ks_1}\psi_{ks_2}\\ \notag
&+&\dfrac{1}{4}\sum\limits_{l_1\ne s_1, l_2\ne s_2}\sum\limits_{j=1}^n\overline{\psi}_{jl_1}\overline{\psi}_{js_1}
\psi_{jl_2}\psi_{js_2}\cdot\sum\limits_{k=1}^n\psi_{kl_1}\psi_{ks_1}
\overline{\psi}_{kl_2}\overline{\psi}_{ks_2}
\end{eqnarray}
and for $p\ge 3$ we have
\begin{equation*}
\begin{array}{c}
\sigma_{p}=\dfrac{1}{(2p)!}\sum\limits_{j<k}((\chi^+_{j,k})^{2p}+(-1)^p(\chi^-_{j,k})^{2p})+
\dfrac{2}{(2p)!}\sum\limits_{j<k}(\chi^+_{j,j})^{2p}\\
=\sum\limits_{l_1,..,l_{2p}=1}^{2m}\sum\limits_{s=0}^{[\frac{p}{2}]}c_{s,l}^{(p)}\sum\limits_{j=1}^n
\overline{\psi}_{jl_1}
..\overline{\psi}_{jl_{p+2s}}\psi_{jl_{p+2s+1}}.. \psi_{jl_{2p}}
\cdot\sum\limits_{k=1}^n\psi_{kl_1}..\psi_{kl_{p+2s}}\overline{\psi}_{kl_{p+2s+1}}..\overline{\psi}_{kl_{2p}},
\end{array}
\end{equation*}
where $c_{s,l}^{(p)}$ are $n$-independent positive coefficients and $l=(l_1,...,l_{2p})$.
Using (\ref{gauss}) we have
\begin{equation}\label{s2_m}
\begin{array}{c}
e^{n^{-2}\kappa_4\sigma_2}=C_2^\prime\displaystyle\int
e^{-|\kappa_4|\left(2\sum\limits_{l_1<l_2<s_1<s_2}\overline{w}_{l_1l_2s_1s_2}w_{l_1l_2s_1s_2}
+\sum\limits_{l_1\ne s_1, l_2\ne s_2}\overline{v}_{l_1l_2s_1s_2}v_{l_1l_2s_1s_2}\right)}\\
\displaystyle\prod\limits_{j=1}^n
e^{\frac{\varepsilon(\kappa_4)}{2n}\sum\limits_{a_1\ne b_1, a_2\ne b_2}\left(v_{a_1a_2b_1b_2}
\overline{\psi}_{ja_1}\overline{\psi}_{jb_1}\psi_{ja_2}\psi_{jb_2}+
\overline{v}_{a_1a_2b_1b_2}\psi_{ja_1}\psi_{jb_1}
\overline{\psi}_{ja_2}\overline{\psi}_{jb_2}\right)}\\
\displaystyle\prod\limits_{j=1}^ne^{\frac{2\varepsilon(\kappa_4)}{n}\sum\limits_{c_1<c_2<d_1<d_2}\left(w_{c_1c_2d_1d_2}
\overline{\psi}_{jc_1}\overline{\psi}_{jc_2}\overline{\psi}_{jd_1}
\overline{\psi}_{jd_2}+\overline{w}_{c_1c_2d_1d_2}\psi_{jc_1}\psi_{jc_2}\psi_{jd_1}
\psi_{jd_2}\right)}
d\,W\,d\,V,
\end{array}
\end{equation}
where
\begin{equation}\label{dVdW}
\begin{array}{c}
d\,W=\displaystyle\prod\limits_{l_1<l_2<s_1<s_2}d\,\Re w_{l_1l_2s_1s_2}\,d\,\Im w_{l_1l_2s_1s_2},\\
d\,V=\displaystyle\prod\limits_{l_1\ne s_1, l_2\ne s_2}d\,\Re v_{l_1l_2s_1s_2}\,d\,\Im v_{l_1l_2s_1s_2},\quad
C_2^\prime=\left(\dfrac{\pi}{2|\kappa_4|}\right)^{-\binom{2m}{4}}
\left(\dfrac{\pi}{|\kappa_4|}\right)^{-(2m)^2(2m-1)^2}.
\end{array}
\end{equation}
As well, (\ref{gauss}) yields for $p\ge 3$
\begin{equation}\label{s3_m}
\begin{array}{c}
\exp\left\{\dfrac{\kappa_{2p}}{n^p}\sigma_p\right\}=C_p^{\prime}\displaystyle\int
\exp\left\{-|\kappa_{2p}|\left(\sum\limits_{l_1,..,l_{2p}=1}^{2m}\sum\limits_{s=0}^{[\frac{p}{2}]}
\overline{r}_{l,s}r_{l,s}\right)\right\}\\
\displaystyle\prod\limits_{j=1}^n\exp\left\{\dfrac{\varepsilon(\kappa_{2p})}{n^{p/2}}\sum\limits_{l_1,..,l_{2p}=1}^{2m}
\sum\limits_{q=0}^{[\frac{p}{2}]}(c_{l,q}^{(p)})^{1/2}(r_{l,q}\overline{\psi}_{jl_1}
..\overline{\psi}_{jl_{p+2q}}\psi_{jl_{p+2q+1}}.. \psi_{jl_{2p}}\right.\\
\left.+\overline{r}_{l,q}
\psi_{jl_1}..\psi_{jl_{p+2q}}\overline{\psi}_{jl_{p+2q+1}}..\overline{\psi}_{jl_{2p}})\right\}d\,R
\end{array}
\end{equation}
with $l=(l_1,\ldots,l_{2p})$ and
\begin{equation}\label{dR}
d\,R=\prod\limits_{l_1,..,l_{2p}=1}^{2m}\prod\limits_{s=0}^{[\frac{p}{2}]}
d\,\Re r_{l,s}d\,\Im r_{l,s},
\quad
C_p^{\prime}=\left(\dfrac{\pi}{|\kappa_{2p}|}\right)^{-[\frac{p}{2}](2m)^{2p}},\,\,p\ge 3.
\end{equation}
Substituting (\ref{s2_m}) -- (\ref{s3_m}) and (\ref{s1_2}) with sums from 1 to $2m$ instead of from $1$ to $2$
in (\ref{ne_fact_m}) and using (\ref{G_Gr}) -- (\ref{G_Gr_1}) we can integrate
over Grassmann variables in (\ref{ne_fact_m}). We get
\begin{equation}\label{F_int1_m_1}
D^{-1}_{2m}F_{2m}(\Lambda)=Z_m\displaystyle\int\limits_{\mathcal{H}_{2m}}d\,Q\displaystyle\int
d\,V\,d\,R\,d\,W\, e^{-\frac{n}{2}\hbox{tr}\,Q^2}
\widetilde{\nu}_n(v,w,r)\Phi^n(iQ+\Lambda,v,w,r),\\
\end{equation}
where $\mathcal{H}_{2m}$ is the space of hermitian $2m\times 2m$ matrices,
\begin{eqnarray}\label{vwr}
v&=&\{v_{a_1a_2b_1b_2}|a_1\ne b_1,\,a_2\ne b_2,\,a_1,a_2,b_1,b_2=1,..,2m\},\\ \notag
w&=&\{w_{a_1a_2b_1b_2}|a_1<a_2<b_1<b_2,\,a_1,a_2,b_1,b_2=1,..,2m\},\\ \notag
r_p&=&\{r_{l,s}|l_1,..,l_{2p}=1,..,2m,\, s=0,..,[p/2]\},\\ \notag
r&=&(r_3,\ldots,r_{2m}),
\end{eqnarray}
and
\begin{equation}\label{nu_til}
\widetilde{\nu}_n(v,w,r)=\exp\left\{-|\kappa_4|\overline{v}v-2|\kappa_4|\overline{w}w-\sum\limits_{p=3}^{2m}|\kappa_{2p}|
\overline{r}_pr_p\right\}.
\end{equation}
$d\,Q$, $d\,V$, $d\,R$ and $d\,W$ are defined in (\ref{dH}) for $n=2m$, (\ref{dVdW}) and (\ref{dR}), and
\begin{equation}\label{Q_m}
Q=
\left(
\begin{array}{cccccc}
\tau_1&u_{12}&u_{13}& .. &u_{1,2m-1}&u_{1,2m}\\
\overline{u}_{12}&\tau_2&u_{23}&..&u_{2,2m-1}&u_{2,2m}\\
\overline{u}_{13}&\overline{u}_{23}&\tau_3&..&u_{3,2m-1}&u_{3,2m}\\
..&..&..&..&..&..\\
\overline{u}_{1,2m-1}&\overline{u}_{2,2m-1}&\overline{u}_{3,2m-1}&..&\tau_{2m-1}&u_{2m,2m-1}\\
\overline{u}_{2m,1}&\overline{u}_{2m,2}&\overline{u}_{2m,3}&..&\overline{u}_{2m-1,2m}&\tau_{2m}\\
\end{array}
\right)
\end{equation}
is obviously hermitian. We denote also
\begin{equation}\label{lam,z}
\begin{array}{c}
\Lambda=\hbox{diag}\{\lambda_1,\lambda_2,\ldots,\lambda_{2m}\},\quad
Z_m=D^{-1}_{2m}\dfrac{n^{2m^2}}{2^m\pi^{2m^2}}\prod\limits_{p=2}^{2m}C_p^\prime.\\
\end{array}
\end{equation}
According to (\ref{G_Gr}) -- (\ref{G_Gr_1}) $\Phi(iQ+\Lambda,v,w,r)$ in (\ref{F_int1_m_1}) is a polynomial of the
entries of $iQ+\Lambda$ and of $\{v_{l_1l_2s_1s_2}/n\}$, $\{w_{l_1l_2s_1s_2}/n\}$, $\{r_{l_1,..,l_{2p},s}/n^{p/2}\}$ with $n$-independent
coefficients and degree at most $2m$ and
\begin{enumerate}

\item  the degree of each variable in $\Phi(iQ+\Lambda,v,w,r)$
is at most one;

\item  $\Phi(iQ+\Lambda,v,w,r)$ does not contain terms $C(iQ+\Lambda)w_{l_1l_2s_1s_2}/n$ or $C(iQ+\Lambda)\overline{w}_{l_1l_2s_1s_2}/n$,
since the terms $\overline{\psi}_{jl_1}\overline{\psi}_{jl_2}\overline{\psi}_{js_1}\overline{\psi}_{js_2}$ or
$\psi_{jl_1}\psi_{jl_2}\psi_{js_1}\psi_{js_2}$ cannot be completed to $\prod\limits_{l=1}^{2m}
\overline{\psi}_{jl}\psi_{jl}$ only by terms $\overline{\psi}_{jl}\psi_{js}$;

\item  $\Phi(iQ+\Lambda,v,w,r)$ can be written as
 \begin{equation}\label{Phi}
\Phi(iQ+\Lambda,v,w,r)=\mdet(iQ+\Lambda)-
\dfrac{2\varepsilon(\kappa_4)}{n}\sigma_1^\prime+\widetilde{f}_n(iQ+\Lambda,v/n,w/n,r_p/n^{p/2}),
\end{equation}
where $\widetilde{f}_n(iQ+\Lambda,v/n,w/n,r/n^{p/2})$ contains all terms of $\Phi(iQ+\Lambda,v,w,r)$
which are $O(n^{-3/2})$ as $n\to\infty$ and as $Q$, $v$, $w$, $r$ are fixed,
and $\sigma_1^\prime$ contains linear with respect to $v$ terms. In view of (\ref{G_Gr_1})
\begin{equation}\label{sig'}
\begin{array}{c}
\sigma_1^\prime=\sum\limits_{l_1\ne s_1,l_2\ne
s_2}(v_{l_1l_2s_1s_2}q_{l_1,s_1,l_2,s_2}+\overline{v}_{l_1l_2s_1s_2}q_{l_2,s_2,l_1,s_1}),
\end{array}
\end{equation}
where $q_{s,l,p,r}$ is $(2m-2)\times(2m-2)$ minor of the matrix $iQ+\Lambda$ without rows
with numbers $s$ and
$l$ and columns with numbers $p$ and $r$.
\end{enumerate}
Recall that we are interested in $\Lambda=\Lambda_0+\widehat{\xi}/n\rho_{sc}(\lambda_0)$, where
$\Lambda_0=\hbox{diag}\{\lambda_0,\ldots,\lambda_0\}$ and
$\widehat{\xi}=\hbox{diag}\{\xi_1,\ldots,\xi_{2m}\}$.
Shift now $\tau_j-i\lambda_0/2-i\xi_j/n\rho_{sc}(\lambda_0)\to \tau_j$, $j=1,..,2m$.
Then similarly to (\ref{F_int1_2}) we obtain in new variables
\begin{eqnarray}\label{F_int1_m}
D^{-1}_{2m}F_{2m}(\Lambda)&=&Z_m\displaystyle\int\limits_{\mathcal{H}_{2m}}d\,Q\displaystyle\int
\widetilde{\nu}_n(v,w,r)\Phi^n(iQ+\Lambda_0/2,v,w,r)\\ \notag
&\times&e^{-\frac{n}{2}\hbox{tr}\,(Q+\frac{i\Lambda_0}{2})^2-
i\hbox{tr}\,(Q+\frac{i\Lambda_0}{2})\widehat{\xi}/\rho_{sc}(\lambda_0)-
\frac{1}{2n}\hbox{tr}\,\frac{\widehat{\xi}^2}{\rho_{sc}(\lambda_0)^2}}d\,V\,d\,R\,d\,W,
\end{eqnarray}
where $Q$ is the hermitian matrix of (\ref{Q_m}) and $d\,Q$, $d\,V$, $d\,R$ and $d\,W$ are defined in (\ref{dH}) for
$n=2m$, (\ref{dVdW}) and
(\ref{dR}).
The (1) condition of $\Phi$ yields
\begin{equation}\label{ogr_Phi1}
\begin{array}{c}
|\Phi(iQ+\frac{\Lambda_0}{2},v,w,r)|\le \displaystyle\prod\limits_{q,s}(1+C|(iQ+\Lambda_0/2)_{qs}|)
\displaystyle\prod\limits_{l_1\ne s_1, l_2\ne s_2}
\left(1+C\left|\dfrac{v_{l_1l_2s_1s_2}}{n}\right|\right)\\
\times\displaystyle\prod\limits_{a_1<a_2<b_1<b_2}\left(1+C\left|\dfrac{w_{a_1a_2b_1b_2}}{n}\right|\right)
\displaystyle\prod\limits_{p=3}^{2m}\displaystyle\prod\limits_{l_1,..,l_{2p}=1}^{2m}\prod\limits_{s=0}^{[\frac{p}{2}]}
\left(1+C\left|\dfrac{r_{l,s}}{n^{p/2}}\right|\right)
\end{array}
\end{equation}
with $n$-independent $C$, and (\ref{Phi}) yields
\begin{equation}\label{ogr_Phi2}
\begin{array}{c}
|\Phi(iQ+\Lambda_0,v,w,r)|\le |\mdet(iQ+\Lambda_0/2)|\displaystyle\prod\limits_{l_1\ne s_1, l_2\ne s_2}
\left(1+C(Q)\left|\dfrac{v_{l_1l_2s_1s_2}}{n}\right|\right)\\
\times\displaystyle\prod\limits_{l_1<l_2<s_1<s_2}\left(1+C(Q)\left|\dfrac{w_{l_1l_2s_1s_2}}{n}\right|\right)
\displaystyle\prod\limits_{p=3}^{2m}\prod\limits_{l_1,..,l_{2p}=1}^{2m}\displaystyle\prod\limits_{s=0}^{[\frac{p}{2}]}
\left(1+C(Q)\left|\dfrac{r_{l,s}}{n^{p/2}}\right|\right).
\end{array}
\end{equation}
Here $C(Q)$ is bounded if $a\le |q_l-i\lambda_0/2|\le A$, $l=1,..,2m$ and
$\{q_l\}_{l=1}^{2m}$ are the eigenvalues of $Q$. Note that if $|\lambda_0|>\delta>0$,
then $|q_l-i\lambda_0/2|\ge \delta^2$ everywhere. Denote
\begin{eqnarray}\label{Om_n_m}
\Omega_n&=&\{(Q,v,w,r): a\le|q_s-i\lambda_0/2|\le A,\,|v_{l_1l_2s_1s_2}|\le \log n,\\ \notag
& &|w_{l_1l_2s_1s_2}|\le \log n,\,
|r_{l,s}|\le \log n\},\\ \notag
\Omega_n^Q&=&\{Q\in \mathcal{H}_{2m}:a\le|q_s-i\lambda_0/2|\le A,\, s=1,..,2m\}.
\end{eqnarray}
According to (\ref{ogr_Phi1}) the integral in (\ref{F_int1_m}) over the domain $\max\limits_{l=1,..,2m}
|q_l-i\lambda_0/2|\ge A$ is $O(e^{-nA^2/4})$, $A\to\infty$ and the integral over the domain $\min\limits_{l=1,..,2m}
|q_l-i\lambda_0/2|\le a$ is $O(e^{-n\log a^{-1}})$, $a\to 0$. Moreover, the bound (\ref{ogr_Phi2}) implies
that this integral over the domain,
where the absolute value of at least one of $\{v_{l_1l_2s_1s_2}\}$, $\{w_{l_1l_2s_1s_2}\}$ or $\{r_{l,s}\}$ is
greater then $\log n$ but $a\le |q_s-i\lambda_0/2|\le A$, $s=1,..,2m$, can be bounded by
$e^{-c\log^2 n}$ (similarly to (\ref{bound_vne})). Therefore, using (\ref{D_int}), (\ref{D_2m}), and (\ref{lam,z})
to bound the integral with $|\mu_n(Q)|$, we can write
\begin{equation}\label{F_int2_m}
\begin{array}{c}
D^{-1}_{2m}F_{2m}(\Lambda)=Z_m\displaystyle\int\limits_{\Omega_n} \mu_n(Q)e^{-\hbox{tr}\,
(Q+\frac{i\Lambda_0}{2})\frac{i\widehat{\xi}}{\rho_{sc}(\lambda_0)}
-2\varepsilon(\kappa_4)\mdet^{-1}(iQ+\Lambda_0/2)\sigma_1^\prime-\frac{1}{2n}\hbox{tr}\,\frac{\widehat{\xi}^2}{\rho_{sc}(\lambda_0)^2}}\\
\times\widetilde{\nu}_n(v,w,r)\left(1+
f_n(Q,v,w,r)\right)d\,Q\,d\,V\,d\,R\,d\,W+O(e^{-c\log^2 n}),
\end{array}
\end{equation}
where $\mu_n$, $\sigma_1^\prime$ and $\widetilde{\nu}_n(v,w,r)$ are defined in (\ref{mu}), (\ref{sig'}) and
(\ref{nu_til}) respectively, and
\begin{equation}\label{f_n}
f_n(Q,v,w,r)=e^{2\varepsilon(\kappa_4)\mdet^{-1}(iQ+\frac{\Lambda_0}{2})\sigma_1^\prime}
\left(1+\dfrac{\widetilde{f}_n(iQ+\frac{\Lambda_0}{2},v,w,r)
-2\varepsilon(\kappa_4)\sigma_1^\prime/n}{\mdet(iQ+\Lambda_0/2)}\right)^n-1.
\end{equation}
Note that $f_n$ is an analytic function of the entries of $Q$, and in view of (\ref{ogr_Phi2}) we have on $\Omega_n$
\begin{equation}\label{bound_f}
|f_n(Q,v,w,r)|\le n^{-1/2}\log^k n,
\end{equation}
where $k$ is independent of $n$. It is easy to check that
\begin{eqnarray}\label{int_mu}
I&:=&\displaystyle\int\limits_{\Omega_n}\nu_n(Q,v,w,r)d\,V\,d\,R\,d\,W\\ \notag
&=&\prod\limits_{p=2}^{2m}(C_p^\prime)^{-1}e^{\kappa_4
\sigma(iQ+\Lambda_0/2)\mdet^{-2}(iQ+\Lambda_0/2)}+O(e^{-c\log^2n}),
\end{eqnarray}
where
\begin{equation}\label{sigma_u_m}
\begin{array}{c}
\nu_n(Q,v,w,r)=\exp\{-2\varepsilon(\kappa_4)\mdet^{-1}(iQ+\Lambda_0/2)\sigma_1^\prime\}\widetilde{\nu}_n(v,w,r),\\
\sigma(iQ+\Lambda_0/2)=\sum\limits_{l_1\ne s_1,l_2\ne s_2}q_{l_1,s_1,l_2,s_2}q_{l_2,s_2,l_1,s_1}\quad
\end{array}
\end{equation}
with $q_{l_1,s_1,l_2,s_2}$ defined in (\ref{sig'}) (but for the matrix $iQ+\Lambda_0/2$ instead of $iQ+\Lambda_0$).
Note also that according to the Cauchy-Binet formula (see \cite{Gant:59}),
we have that $\sigma(iQ+\Lambda_0/2)$ is the sum $S_{2m-2}(A)$ of principal minors of order
$(2m-2)\times (2m-2)$
of the matrix $$A=(iQ^*+\Lambda_0/2)(iQ+\Lambda_0/2)=U^*(iT_0+\Lambda_0/2)^2U,$$
where $U$ is a unitary $2m\times 2m$ matrix diagonalizing $Q$ and
$T_0=\hbox{diag}\{q_1,..,q_{2m}\}$, i.e. $Q=U^{*}T_0 U$.
Since $S_{2m-2}(A)$ is a coefficient under $\lambda^2$ in the characteristic polynomial $\mdet(A-\lambda I)$,
$S_{2m-2}(A)$
is unitary invariant, and thus $\sigma(iQ+\Lambda_0/2)$ is unitary invariant too.
Therefore, we have on $\Omega_n$ of (\ref{Om_n_m})
\[
\left|\sigma(iQ+\frac{\Lambda_0}{2})\mdet^{-2}(iQ+\frac{\Lambda_0}{2})\right|=
\left|\sum\limits_{1\le s<l\le 2m}\dfrac{1}{(iq_s+\frac{\lambda_0}{2})^2(iq_l+\frac{\lambda_0}{2})^2}\right|\le
C,
\]
and hence $|I|>C>0$.
This, (\ref{F_int2_m}) and (\ref{int_mu})  yield
\begin{eqnarray}\label{F_int3_m}
D^{-1}_{2m}F_{2m}(\Lambda)&=&\dfrac{n^{2m^2}D_{2m}^{-1}}{2^m\pi^{2m^2}}\int\limits_{\Omega_n^Q}
e^{-\hbox{tr}\,(Q+\frac{i\Lambda_0}{2})\frac{i\widehat{\xi}}{\rho_{sc}(\lambda_0)}+
\kappa_4\,\sigma(iQ+\frac{\Lambda_0}{2})/\mdet(iQ+\frac{\Lambda_0}{2})^2}
\\ \notag
&\times&\mu_n(Q)\left(1+f^{(1)}_n(Q)\right)d\,Q+O(e^{-c\log^2 n}),
\end{eqnarray}
where $\mu_n$ is defined in (\ref{mu}) and
\begin{multline}\label{f-1}
f^{(1)}_n(Q)=e^{-\frac{1}{2n}\hbox{tr}\,\frac{\widehat{\xi}^2}{\rho_{sc}(\lambda_0)^2}}-1\\
+I^{-1}e^{-\frac{1}{2n}\hbox{tr}\,\frac{\widehat{\xi}^2}{\rho_{sc}(\lambda_0)^2}}\displaystyle\int\limits_{\Omega_n}\nu_n(Q,v,w,r)f_n(Q,v,w,r)d\,V\,d\,R\,d\,W
\end{multline}
with $I$ of (\ref{int_mu}) and $\nu_n$ of (\ref{sigma_u_m}). According to (\ref{f_n}) and bound from below of $|I|$ on $\Omega_n$, we have
\begin{equation}\label{f_bound}
|f^{(1)}_n(Q)|\le \log^k n/n^{1/2}.
\end{equation}
Besides, $f^{(1)}_n(Q)$ is analytic in elements of $Q$.

Let us change variables to $Q=U^{*}T U$, where $U$ is a unitary $2m\times 2m$ matrix and
$T=\hbox{diag}\{t_1,\ldots,t_{2m}\}$. The differential $d\, Q$ in (\ref{F_int3_m}) transforms
to $\triangle^2(T)d\,T d\,\mu(U)$, where
$d\,T=\prod\limits_{l=1}^{2m}d\,t_l$, $\triangle(T)$ is a Vandermonde determinant of $\{t_l\}_{l=1}^{2m}$,
and $\mu(U)$ is the normalized to unity Haar measure on the unitary group
$U(2m)$ (see e.g. \cite{Me:91}, Section 3.3). Functions $\mdet (iQ+\frac{\Lambda_0}{2})$,
$\hbox{tr}\,(Q+\frac{i\Lambda_0}{2})^2$ and $\sigma(iQ+\frac{\Lambda_0}{2})$
(as we proved before) are unitary invariant. Hence, (\ref{F_int3_m}) implies
\begin{equation}\label{F_int4_m}
\begin{array}{c}
D^{-1}_{2m}F_{2m}(\Lambda)=\dfrac{(-1)^{mn}n^{2m^2}}{D_{2m}2^m\pi^{2m^2}}\int\limits_{U(2m)}d\mu(U)
\int\limits_{(L_a^A)^{2m}}d\,T \prod\limits_{l=1}^{2m}(t_l-\frac{i\lambda_0}{2})^n\\
\times e^{-\frac{n}{2}\sum\limits_{l=1}^{2m}(t_l+\frac{i\lambda_0}{2})^2-
\hbox{tr}\,U^*(T+\frac{i\Lambda_0}{2})U\frac{i\widehat{\xi}}{\rho_{sc}(\lambda_0)}+\kappa_4\sum\limits_{1\le l<s\le 2m}
(t_l-\frac{i\lambda_0}{2})^{-2}(t_s-\frac{i\lambda_0}{2})^{-2}}\\
\times\triangle^2(T)
\left(1+f^{(1)}_n(U^*TU)\right) +O(e^{-c\log^2 n}),
\end{array}
\end{equation}
where $L_a^A$ is defined in (\ref{L}).
Using Proposition \ref{p:Its-Z} we have
\begin{equation}\label{F_int5_m}
\begin{array}{c}
D^{-1}_{2m}F_{2m}(\Lambda)=\dfrac{(-1)^{mn}n^{2m^2}}{D_{2m}2^m\pi^{2m^2}}
\int\limits_{(L_a^A)^{2m}}
e^{-\frac{n}{2}\sum\limits_{l=1}^{2m}(t_l+\frac{i\lambda_0}{2})^2+
\kappa_4\sum\limits_{1\le i<j\le 2m}(t_i-\frac{i\lambda_0}{2})^{-2}(t_j-\frac{i\lambda_0}{2})^{-2}}\\
\triangle^2(T)\prod\limits_{l=1}^{2m}(t_l-\frac{i\lambda_0}{2})^n
\left(\dfrac{\pi^me^{-i\sum\limits_{l=1}^{2m}(t_l+\frac{i\lambda_0}{2})\frac{\xi_l}{\rho_{sc}(\lambda_0)}}}
{\Delta(T)\Delta(-i\widehat{\xi}/\rho_{sc}(\lambda_0))}+\widetilde{f}^{(2)}_n(T)\right)
d\,T+O(e^{-c\log^2 n}),
\end{array}
\end{equation}
where
\[
\widetilde{f}^{(2)}_n(T)=\int \exp\left\{-i\hbox{tr}\,U^*(T+\frac{i\Lambda_0}{2})U\frac{\widehat{\xi}}{\rho_{sc}(\lambda_0)}\right\}f^{(1)}_n(U^*TU)d\,\mu(U).
\]
According to (\ref{f_bound}), we get that
\begin{equation}\label{f2_bound}
|\widetilde{f}^{(2)}_n(T)|\le n^{-1/2}\log^k n, t_l\in L_a^A,\,\,l=1,..2m.
\end{equation}
Hence, we obtain finally
\begin{equation}\label{F_int_m}
\begin{array}{c}
D^{-1}_{2m}F_{2m}(\Lambda)
=\dfrac{(-1)^{mn}n^{2m^2}}{D_{2m}2^m\pi^{2m^2-m}}\int\limits_{(L_a^A)^{2m}}
\prod\limits_{l=1}^{2m}(t_l-\frac{i\lambda_0}{2})^n \dfrac{\triangle(T)}
{\Delta(\widehat{\xi})}(i\rho_{sc}(\lambda_0))^{m(2m-1)}\\
\times e^{-\frac{n}{2}\sum\limits_{l=1}^{2m}(t_l+\frac{i\lambda_0}{2})^2-i\sum\limits_{l=1}^{2m}
(t_l+\frac{i\lambda_0}{2})\frac{\xi_l}{\rho_{sc}(\lambda_0)}+
 \kappa_4\sum\limits_{l_1<l_2}(t_{l_1}-\frac{i\lambda_0}{2})^{-2}
(t_{l_2}-\frac{i\lambda_0}{2})^{-2}}\\
\times\left(1+f^{(2)}_n(T)\right)\prod\limits_{j=1}^{2m}d\,t_j+O(e^{-c\log^2 n}),
\end{array}
\end{equation}
where
\[
f^{(2)}_n(T)=\Delta(T)\Delta(-i\widehat{\xi}/\rho_{sc}(\lambda_0))
e^{i\sum\limits_{l=1}^{2m}t_l\xi_l/\rho_{sc}(\lambda_0)}\widetilde{f}^{(2)}_n(T).
\]
$f^{(2)}_n(T)$ is an analytic function bounded by $n^{-1/2}\log^kn$ if $t_l\in L_a^A$, $l=1,..,2m$.

\section{Asymptotic analysis.}
In this section we prove Theorem \ref{thm:1} passing to the limit $n\to\infty$ in (\ref{F_int_m}) for
$\lambda_j=\lambda_0+\xi_j/n\rho_{sc}(\lambda_0)$, where $\rho_{sc}$ is defined in
(\ref{rho_sc}), $\lambda_0\in(-2,2)$ and $\xi_j\in [-M,M]\subset\mathbb{R}$, $j=1,..,2m$.

To this end consider the function
\begin{equation}\label{V}
V(t,\lambda_0)=\dfrac{t^2}{2}+\frac{i\lambda_0}{2} t-\log (t-i\lambda_0/2)-\dfrac{4-\lambda_0^2}{8}.
\end{equation}
Then (\ref{F_int_m}) yields
\begin{equation}\label{F_int_V}
\dfrac{D_{2m}^{-1}}{(n\rho_{sc}(\lambda_0))^{m^2}}F_{2m}(\Lambda)=Z_{m,n}
\intd\limits_{(L_a^A)^{2m}} W_n(t_1,\ldots,t_{2m})d\,T+O(e^{-c\log^2 n}),
\end{equation}
where $D_{2m}$ is defined in (\ref{D_2m}),
\begin{equation}\label{W_n}
\begin{array}{c}
W_n(t_1,\ldots,t_{2m})=
e^{-n\sum\limits_{l=1}^{2m}V(t_l,\lambda_0)-i\sum\limits_{l=1}^{2m}\frac{\xi_l}{\rho_{sc}(\lambda_0)} t_l}\dfrac{\triangle(T)}
{\triangle(\widehat{\xi})}\\
\times e^{\kappa_4
\sum\limits_{1\le l<s\le 2m}(t_l-\frac{i\lambda_0}{2})^{-2}(t_s-\frac{i\lambda_0}{2})^{-2}}
\left(1+f^{(2)}_n(T)\right),
\end{array}
\end{equation}
and
\begin{equation}\label{Z_mn}
Z_{m,n}=\dfrac{(-1)^{mn}n^{m^2}\rho_{sc}(\lambda_0)^{m(m-1)}e^{-m\kappa_4}
}{(-i)^{m(2m-1)}2^{2m}\pi^{2m^2}}.
\end{equation}
Now we need
\begin{lemma}\label{l:min_L}
The function $\Re V(t,\lambda_0)$ for $t\in \mathbb{R}$ has the minimum at the points
\begin{equation}\label{x_pm}
t=x_{\pm}:=\pm\dfrac{\sqrt{4-\lambda_0^2}}{2}.
\end{equation}
Moreover, if $t\not\in U_n(x_\pm):=(x_\pm-n^{-1/2}\log n,
x_\pm+n^{-1/2}\log n)$, then we have for sufficiently big $n$
\begin{equation}\label{ineqv_ReV}
\Re V(t,\lambda)\ge \dfrac{C\log^2n}{n}.
\end{equation}
\end{lemma}
\begin{proof}
Note that for $t\in\mathbb{R}$
\begin{equation}\label{ReV}
\Re V(t,\lambda_0)=\dfrac{1}{2}\left(t^2-(4-\lambda_0^2)/4-\log (t^2+\lambda_0^2/4)\right),
\end{equation}
and thus
\begin{eqnarray}\label{ReV_pr}
\dfrac{d}{d\,t}\Re V(t,\lambda_0)&=&t-\dfrac{t}{t^2+\lambda_0^2/4},\\ \notag
\dfrac{d^2}{d\,t^2}\Re V(t,\lambda_0)&=&1-\dfrac{1}{t^2+\lambda_0^2/4}+\dfrac{2t^2}{(t^2+\lambda_0^2/4)^2}.
\end{eqnarray}
Therefore, $t=x_\pm$ of (\ref{x_pm}) are the minimum points of
$\Re V(t,\lambda_0)$.
Note that
\begin{eqnarray}\label{V_pm}
V_+:=V(x_{+},\lambda_0)&=&\dfrac{i\lambda_0\sqrt{4-\lambda^2_0}}{4}-i\arcsin(-\lambda_0/2),\\ \notag
V_-:=V(x_{-},\lambda_0)&=&-\dfrac{i\lambda_0\sqrt{4-\lambda^2_0}}{4}+i\arcsin(-\lambda_0/2)
-i\pi.
\end{eqnarray}
Thus we have
\[
\Re V(x_{\pm},\lambda_0)=0.
\]
Expanding $\Re V(t,\lambda_0)$ into the Taylor series, we obtain for $t\in U_n(x_\pm)$,
using (\ref{ReV_pr}) -- (\ref{V_pm})
\begin{equation}\label{ineqv_ne0}
\Re V(t,\lambda_0)=\dfrac{4-\lambda_0^2}{4}(t-x_\pm)^2+O(n^{-3/2}\log^3n),
\end{equation}
where $x_\pm$ is defined in (\ref{x_pm}). Hence, for $t\not\in U_n(x_\pm)$ we get
\[
\Re V(t,\lambda)\ge \dfrac{C\log^2n}{n},
\]
which completes the proof of the lemma.
\end{proof}
Next note that since $|t_j-i\lambda_0/2|>a$ for $t_j\in L_a^A$, $j=1,..,2m$, we have
\begin{equation}\label{ots_exp}
\left|\exp\left\{\kappa_4
\sum\limits_{1\le l<s\le 2m}(t_l-i\lambda_0/2)^{-2}(t_s-i\lambda_0/2)^{-2}\right\}\right|\le C.
\end{equation}
This, the inequality $|\triangle(T)/\triangle(\widehat{\xi})|\le C_1$
for $|t_j|\le A$, $j=1,..,2m$ and distinct $\{\xi_j\}_{j=1}^{2m}$, (\ref{ineqv_ReV}) and (\ref{ots_exp}) yield
\begin{equation*}
\left|Z_{m,n}\int\limits_{L_a^A\setminus(U_{+}\cup U_{-})}\int\limits_{L_a^A}..\int\limits_{L_a^A}
W_n(t_1,\ldots,t_{2m})d\,T
\right|\le C_1n^{m^2}e^{-C_2\log^2n},
\end{equation*}
where $L_{a}^A$, $W_n$ and $Z_{m,n}$ are defined in (\ref{L}), (\ref{W_n}) and
(\ref{Z_mn}) respectively, and
\begin{equation}\label{U_pm}
U_{\pm}=\{t\in \mathbb{R}: |t-x_\pm|\le n^{-1/2}\log n\}
\end{equation}
with $x_\pm$ of (\ref{x_pm}).

  Note that we have for $t\in U_{\pm}$ in view of (\ref{V}) and (\ref{V_pm}) as $n\to\infty$
\begin{equation}\label{as_razl_V}
V(t,\lambda_0)=V_{\pm}+\left(1+\dfrac{1}{(x_{\pm}-i\lambda_0/2)^2}\right)
\dfrac{(t-x_{\pm})^2}{2}+
f_\pm(t-x_{\pm}),
\end{equation}
where $f_\pm(t-x_{\pm})=O((t-x_{\pm})^3)$.
%
Shifting $t_j-x_\pm\to t_j$ for $t_j\in U_{\pm}$ we obtain using (\ref{V_pm}) that the r.h.s. of
(\ref{F_int_V}) can be rewritten as
\begin{equation}\label{int_okr1}
\begin{array}{c}
Z_{m,n}\sum\limits_{\alpha}
\intd\limits_{(U_{n})^{2m}}\prod\limits_{j=1}^{2m}d\,t_j\prod\limits_{j=1}^{2m}
e^{-\frac{nc_{\alpha_j}}{2}t_j^2-\frac{i\xi_jt_j}{\rho_{sc}(\lambda_0)}-nf_{\alpha_j}(t_j)}\dfrac{\triangle(t_1+x_{\alpha_1},\dots,t_{2m}+x_{\alpha_{2m}})}
{\triangle(\xi_1,\dots,\xi_{2m})}\\
e^{\kappa_4\sum\limits_{1\le l<s\le 2m}(t_l+p_{\alpha_l})^{-2}
(t_s+p_{\alpha_s})^{-2}-\sum\limits_{j=1}^{2m}
(nV_{\alpha_j}+i x_{\alpha_j}\xi_j/\rho_{sc}(\lambda_0))}(1+f_n^{(2)}(T))+O(e^{-c\log^2 n}),
\end{array}
\end{equation}
where sum is over all collection $\alpha=\{\alpha_j\}_{j=1}^{2m}$, $\alpha_j=\pm$, $j=1,..,2m$ and
\begin{equation}\label{cp}
c_{\pm}=1+p_{\pm}^{-2},\quad p_\pm=x_\pm-i\lambda_0/2,\quad U_{n}=(-2n^{-1/2}\log n, 2n^{-1/2}\log n).
\end{equation}
Note that
\begin{equation}\label{I}
\begin{array}{c}
I:=\displaystyle\int\limits_{(U_{n})^{2m}}
e^{-\sum\limits_{j=1}^{2m}\frac{nc_{\alpha_j}}{2}t_j^2-\sum\limits_{j=1}^{2m}\frac{i\xi_jt_j}{\rho_{sc}(\lambda_0)}-
\sum\limits_{j=1}^{2m}nf_{\alpha_j}(t_j)}
\triangle(t_1+x_{\alpha_1},\dots,t_{2m}+x_{\alpha_{2m}})\prod\limits_{j=1}^{2m}d\,t_j\\
=\mdet\left\{\displaystyle\int\limits_{U_{n,j}}\left(t_j+x_{\alpha_j}-
\dfrac{i\xi_j}{n\rho_{sc}(\lambda_0)c_{\alpha_j}}\right)^{k-1}e^{-\frac{nc_{\alpha_j}}{2}t_j^2-
nf_{\alpha_j}(t_j-\frac{i\xi_j}{n\rho_{sc}(\lambda_0)c_{\alpha_j}})}d\,t_j\right\}_{j,k=1}^{2m},
\end{array}
\end{equation}
where $$U_{j,n}=\left(-2n^{-1/2}\log n+\frac{i\xi_j}{n\rho_{sc}(\lambda_0)c_{\alpha_j}},
2n^{-1/2}\log n+\frac{i\xi_j}{n\rho_{sc}(\lambda_0)c_{\alpha_j}}\right).$$
Since $f_\pm(t)=O(t^3)$, changing variables to $\sqrt{n}t_j\to t_j$, expanding $\exp\{-nf_{\alpha_j}
(t_j/{\sqrt{n}}-i\xi_j/n\rho_{sc}(\lambda_0)c_{\alpha_j})\}$ in (\ref{I}), and keeping the terms up to the order
$n^{-4m^2}$, we obtain as $n\to\infty$
\begin{equation}\label{I_det}
I=\prod\limits_{j=1}^{2m}\sqrt{\dfrac{2\pi}{nc_{\alpha_j}}}\,\,\mdet\left\{\left(x_{\alpha_j}-
\dfrac{i\xi_j}{n\rho_{sc}(\lambda_0)c_{\alpha_j}}\right)^{k-1}+
\dfrac{1}{n}P^{(\alpha_j)}_{k,m}(\xi_j/n)\right\}_{j,k=1}^{2m}(1+o(1)),
\end{equation}
where $P^{(+)}_{k,m}$ and $P^{(-)}_{k,m}$ are polynomials with $n$- and $j$-independent (but $k$-dependent)
coefficients of degree at most $4m^2$. Consider
\begin{equation}\label{det}
D(\xi/n,\lambda)=\mdet\left\{\left(x_{\alpha_j}-
\dfrac{i\xi_j}{n\rho_{sc}(\lambda_0)c_{\alpha_j}}\right)^{k-1}+\lambda P^{(\alpha_j)}_{k,m}(\xi_j/n)
\right\}_{j,k=1}^{2m}.
\end{equation}
Note that $D(\xi/n,\lambda)$ is a polynomial of $\{\xi_j/n\rho_{sc}(\lambda_0)\}_{j=1}^{2m}$ and $\lambda$.
Without loss of generality, let $\alpha_1=...=\alpha_s=+$, $\alpha_{s+1}=...=\alpha_{2m}=-$.
Then it is easy to see that if $\xi_j=\xi_l$ for $j,l=1,..,s$ or $j,l=s+1,..,2m$, then
$D(\xi/n,\lambda)=0$. Thus,
\begin{equation}\label{det1}
D(\xi/n,\lambda)=\Delta(\xi_1/n,\ldots,\xi_s/n)\\
\Delta(\xi_{s+1}/n,\ldots,\xi_{2m}/n)(C_0+\lambda F(\xi/n,\lambda)),
\end{equation}
where $F(\xi/n,\lambda)$ is a polynomial with bonded coefficients.
Substituting $\lambda=0$ in (\ref{det}) and computing the Vandermonde determinant, we obtain
\begin{equation*}
\begin{array}{c}
C_0=\left(\dfrac{-i}{\rho_{sc}(\lambda_0)c_+}\right)^{\frac{s(s-1)}{2}}
\left(\dfrac{-i}{\rho_{sc}(\lambda_0)c_-}\right)^{\frac{(2m-s)(2m-s-1)}{2}}\\
\prod\limits_{j=1}^{s}\prod\limits_{k=s+1}^{2m}\left(x_{+}-x_-
-
\dfrac{i\xi_j}{n\rho_{sc}(\lambda_0)c_{+}}+\dfrac{i\xi_k}{n\rho_{sc}(\lambda_0)c_{-}}\right)\\
=\left(\dfrac{-i}{\rho_{sc}(\lambda_0)c_+}\right)^{\frac{s(s-1)}{2}}
\left(\dfrac{-i}{\rho_{sc}(\lambda_0)c_-}\right)^{\frac{(2m-s)(2m-s-1)}{2}}(x_+-x_-)^{s(2m-s)}(1+o(1)).
\end{array}
\end{equation*}
Hence, for $\alpha_1=...=\alpha_s=+$, $\alpha_{s+1}=...=\alpha_{2m}=-$ we get from (\ref{det1}) as $n\to\infty$
\begin{equation}\label{I_2}
\dfrac{n^{m^2}I}{\Delta(\widehat{\xi})}\\
=
\dfrac{2^m\pi^m(-i/\rho_{sc}(\lambda_0))^{m(m-1)+(m-s)^2}n^{-(m-s)^2}(\sqrt{4-\lambda_0^2})^{s(2m-s)}(1+o(1))}{(c_+)^{s^2/2}(c_-)^{(2m-s)^2/2}
\prod\limits_{j=1}^s\prod\limits_{l=s+1}^{2m}(\xi_j-\xi_{l})}.
\end{equation}
This expression has the order at most $O(1)$, and for $s\ne m$ it is of order $o(1)$.
Hence, the terms of (\ref{int_okr1}) are
 always of order $O(1)$ and the equality holds only if $m$ of $\{\alpha_j\}_{j=1}^{2m}$ are pluses,
and $m$ last ones are minuses. Consider one of such terms in (\ref{int_okr1}), for example
$\alpha_1=..=\alpha_m=1$, $\alpha_{m+1}=..=\alpha_{2m}=-1$. Substituting the expressions (\ref{Z_mn}), (\ref{cp})
and (\ref{I_2}) with $s=m$ we can rewrite this term as
\begin{equation}\label{int_add1}
\dfrac{e^{m(m-1)\kappa_4(\lambda_0^2-2)^2/2}}
{\pi^{2m^2-2m}}\dfrac{i^{m(m+1)}e^{i\pi(\xi_{m+1}+..+\xi_{2m}-\xi_1-..-\xi_m)}}{(2i\pi)^m\prod\limits_{i,j=1}^m(\xi_i-\xi_{m+j})}
\end{equation}
In view of identity
\[
\dfrac{\mdet\left\{\dfrac{\sin(\pi(\xi_j-\xi_{m+k}))}{\pi(\xi_j-\xi_{m+k})}\right\}_{j,k=1}^m}{\Delta(\xi_1,..,\xi_m)
\Delta(\xi_{m+1},..,\xi_{2m})}=\dfrac{\mdet\left\{\dfrac{e^{i\pi(\xi_j-\xi_{m+k})}-e^{i\pi(\xi_{m+k}-\xi_j)}}{\xi_j-\xi_{m+k}}\right\}_{j,k=1}^m}{
(2i\pi)^m\Delta(\xi_1,..,\xi_m)
\Delta(\xi_{m+1},..,\xi_{2m})},
\]
the determinant in the l.h.s. is the sum of $\exp\{i\pi\sum\limits_{j=1}^{2m}\alpha_j\xi_j\}$ over
the collection $\{\alpha_j\}_{j=1}^{2m}$, in which $m$ of elements are pluses, and $m$ last ones are minuses,
with certain coefficient.
In view of the identity (see \cite{Po-Se:76}, Problem 7.3)
\[
(-1)^{\frac{m(m-1)}{2}}\dfrac{\prod\limits_{k<l}(a_k-a_l)(b_k-b_l)}{\prod\limits_{k,l=1}^m(a_k-b_l)}
=\det\left[\dfrac{1}{a_k-b_j}\right]_{k,j=1}^m.
\]
the coefficient under $\exp\{i\pi(\xi_{m+1}+..+\xi_{2m}-
\xi_1-..-\xi_m)\}$ is equal to
\[
\dfrac{\mdet\left\{\dfrac{1}{\xi_{m+k}-\xi_j}\right\}_{j,k=1}^m}{
(2i\pi)^m\Delta(\xi_1,..,\xi_m)
\Delta(\xi_{m+1},..,\xi_{2m})}=\dfrac{(-1)^{\frac{m(m-1)}{2}}}{(-1)^{m^2}(2i\pi)^m\prod\limits_{i,j=1}^m(\xi_i-\xi_{m+j})}.
\]
Other coefficients can be computed analogously. Thus, restricting the sum in (\ref{int_okr1})
to that over the collection $\{\alpha_j\}_{j=1}^{2m}$, in which $m$ of elements are pluses, and $m$ last ones are minuses,
and using (\ref{int_add1}), we obtain Theorem \ref{thm:1} after certain algebra.

\end{document}